\documentclass[aps,prr,twocolumn,nofootinbib]{revtex4-2}
\usepackage{amssymb,amsmath}
\usepackage{bm,bbm}
\usepackage{graphicx}
\usepackage[bookmarks=false]{hyperref}
\hypersetup{colorlinks=true,citecolor=blue,linkcolor=red,%
urlcolor=blue,pdfstartview=FitH,bookmarksopen=true}
\usepackage[T1]{fontenc}
\usepackage[osf,sc]{mathpazo}
\usepackage{dsfont}
\usepackage{pifont}
\usepackage{comment}
\usepackage{braket}
\usepackage{contour}
\usepackage{bbding}
\usepackage{multirow}
\usepackage{tikz}
\usetikzlibrary{graphs}

\newcommand{\tr}{\mathrm{tr}}
\newcommand{\cC}{\mathcal{C}}
\newcommand{\cG}{\mathcal{G}}
\newcommand{\cP}{\mathcal{P}}

\usepackage{amsthm}
\newtheoremstyle{note}      
  {\topsep/2}              	
  {\topsep/2}            	
  {}                        
  {\parindent}             	
  {\itshape}                
  {.---}                    
  {0pt}                     
  {\thmname{#1}\thmnumber{ \itshape#2}\thmnote{ (#3)}} 

\newtheorem{theorem}{Theorem}

\theoremstyle{definition}
\newtheorem{definition}{Definition}

\theoremstyle{remark}

\begin{document}
\title{Geometric mean of bipartite concurrences as a genuine multipartite entanglement measure}

\author{Yinfei Li}
\affiliation{Key Laboratory of Advanced Optoelectronic Quantum Architecture and Measurement of Ministry of Education, School of Physics, Beijing Institute of Technology, Beijing 100081, China}

\author{Jiangwei Shang}
\email{jiangwei.shang@bit.edu.cn}
\affiliation{Key Laboratory of Advanced Optoelectronic Quantum Architecture and Measurement of Ministry of Education, School of Physics, Beijing Institute of Technology, Beijing 100081, China}

\date{\today}
%

\begin{abstract}
In this work we propose the geometric mean of bipartite concurrences as a genuine multipartite entanglement measure. This measure achieves the maximum value for absolutely maximally entangled states and has desirable properties for quantifying potential quantum resources. The simplicity and symmetry in the definition facilitates its computation for various multipartite entangled states including the GHZ states and the $W$ states. With explicit examples we show that our measure results in distinct entanglement orderings from other measures, and can detect differences in certain types of genuine multipartite entanglement while other measures cannot. These results justify the potential application of our measure for tasks involving genuine multipartite entanglement.
\end{abstract}

\maketitle
%


\section{Introduction}%
Entanglement \cite{RevModPhys.81.865,DetectGUHNE20091} is one of the unique phenomena that distinguishes the quantum world and the classical one, which plays a central role in the advancing of quantum technologies. Most, if not all, quantum information processing tasks regard entanglement as an indispensable resource~\cite{RevModPhys.91.025001}, ranging from quantum communication and quantum enhanced metrology to quantum computing.

In the past few decades, how to quantify entanglement has received extensive attention \cite{PhysRevLett.78.2275,10.1063_1.1494474}. One property was sorted out as a necessity for all entanglement measures, i.e., monotonicity~\cite{PhysRevA.54.3824,doi:10.1080/09500340008244048}, that the amount of entanglement cannot increase under local operations and classical communication (LOCC). A stronger version of this requirement is that entanglement cannot increase on average under stochastic LOCC (SLOCC). For two-qubit entangled states, there exists only a single SLOCC equivalent class, and all the entanglement monotones agree on which state is more entangled~\cite{Singh:20,PhysRevLett.127.040403}.  However, the complexity of multipartite entanglement grows dramatically with the increasing number of parties. Even for three-qubit systems, it is impossible to obtain a consistent entanglement ordering through monotonicity~\cite{PhysRevA.62.062314}. Therefore, the understanding and quantification of multipartite entanglement is still in progress~\cite{PhysRevLett.127.140501}.

The notion of multipartite entanglement can be refined into the so-called genuine multipartite entanglement (GME)~\cite{PhysRevA.65.012107,G_hne_2005,PhysRevA.73.052319}. It is worth noting that genuinely entangled states can excel bipartite entangled states in tasks such as quantum secret sharing~\cite{PhysRevA.59.1829}, measurement-based quantum computation~\cite{Briegel2009} and quantum-enhanced measurements~\cite{doi:10.1126/science.1104149}, and is essential for quantum algorithms to outperform their classical counterparts \cite{doi:10.1098/rspa.2002.1097}. Being capable of generating GME states has been regarded as a benchmark for fault-tolerant quantum computing \cite{PRXQuantum.2.020304} and many-body experiments~\cite{doi.org/10.1038/s42254-018-0003-5,PhysRevLett.117.210502}. Numerous methods have been put forward to detect GME~\cite{2010,PhysRevLett.106.020405,PhysRevLett.104.210501,PhysRevA.83.040301, PhysRevLett.88.170405,PhysRevA.80.022109,PhysRevLett.112.140404}. Beyond witnessing the existence of GME, the utility of a genuinely entangled state also depends on the amount of GME it possesses, which naturally leads to the study of GME measures~\cite{PhysRevA.83.062325}.

\begin{definition}\label{def1}
A GME measure is an entanglement monotone that satisfies the following two conditions:
\begin{enumerate}
  \renewcommand{\theenumi}{(\alph{enumi})}
  \item The measure vanishes for all biseparable states.
  \item The measure is strictly positive for genuinely entangled states.
\end{enumerate}
\end{definition}

Substantial achievements have been made in the research of multipartite entanglement measures~\cite{PhysRevLett.106.190502,PhysRevA.67.012108,PhysRevA.82.032121,doi:10.1063/1.1497700,PhysRevA.69.062311,PhysRevA.79.062308}, but most of these measures do not fully satisfy Definition~\ref{def1}~\cite{PhysRevLett.127.040403}. Worst of all, very few GME measures are available. Here are three representative examples. The geometric measure~\cite{PhysRevA.68.042307} was first considered as the minimal distance between the target state and a product state, then the minimal distance from $k$-separable states was investigated as the generalized geometric measure (GGM)~\cite{PhysRevA.77.062304,PhysRevA.81.012308}, which is a GME measure when ${k=2}$. The GME-concurrence, also known as the genuinely multipartite concurrence (GMC)~\cite{PhysRevA.83.062325,PhysRevA.86.062303}, was introduced as the minimal bipartite concurrence with computable lower bounds. Recently, the concurrence fill~\cite{PhysRevLett.127.040403} was proposed as a three-qubit GME measure, which is in fact equal to the square root of the area of the three-qubit concurrence triangle. Attempts to generalize this measure have also been reported, see for instance Ref.~\cite{guo2021genuine}. However, all these measures cannot fully describe the nature of GME. Both the GGM and GMC rely on minimization arguments, which lose sight of the global distribution of entanglement among the parties. The concurrence fill takes all the bipartite entanglements into account, but is not scalable for more than three qubits.

In this work, we propose a GME measure that is qualified for the faithful certification of GME. Our measure has the virtue of discriminance, strong monotonicity, convexity, scalability, normalization, and smoothness. These features are favorable for acting as a potential quantum resource, and the simplicity in the definition makes it convenient for practical applications. The proof of strong monotonicity and determinance can be applied to the geometric mean of any bipartite entanglement monotones, and thus brings forth a family of GME measures. For certain multipartite pure states, the symmetry in our measure can simplify the computation, so we take the $n$-qubit Greenberger–Horne–Zeilinger (GHZ) states and $W$ states as a demonstration. In addition, we show that our measure leads to distinct entanglement orderings as compared to other GME measures, thus revealing its unique aspects as a GME measure.

\section{Geometric mean of bipartite concurrences}%
Concurrence \cite{WoottersEntof2Q} is one of the most celebrated entanglement monotones defined for bipartite quantum states, which has a direct connection to the entanglement of formation in two-qubit cases \cite{PhysRevLett.80.2245}, and has been extended to higher dimensions~\cite{PhysRevA.64.042315,PhysRevLett.92.167902,PhysRevLett.95.040504,PhysRevA.74.052303}.
The generalized bipartite concurrence for a pure state $\ket{\psi}$ between its subsystems $A$ and $B$ is given by 
\begin{equation}\label{eq1}
  \cC_{AB}(\ket{\psi})=\sqrt{2\bigl(1-\tr\bigl(\rho_A^2\bigr)\bigr)}\,,
\end{equation}
where ${\rho_A=\tr_B{\bigl(\ket{\psi}\!\bra{\psi}\bigr)}}$. A regularized expression, of which we adopt in this work, can be written as
\begin{equation}\label{eq2}
  \cC_{AB}(\ket{\psi})=\sqrt{\frac{d_\text{min}}{d_\text{min}-1}\bigl(1-\tr\bigl(\rho_A^2\bigr)\bigr)}\,,
\end{equation}
where $d_\text{min}$ denotes the dimension of the smaller subsystem.
$\cC_{AB}(\ket{\psi})$ maximizes to $1$ when the reduced density matrix of the smaller subsystem is the maximally mixed state, and minimizes to $0$ when $\ket{\psi}$ is a product state.

Relying on the concept of bipartite concurrence, we propose a GME measure, namely the geometric mean of bipartite concurrences (GBC).
\begin{definition}\label{def2}
For an arbitrary $n$-partite pure state $\ket{\psi}$, the GBC is defined by
\begin{equation}\label{eq4}
    \cG(\ket{\psi})=\sqrt[{c(\alpha)}]{\cP(\ket{\psi})}\,,
\end{equation}
where ${\alpha=\bigl\{\alpha_i\bigr\}}$ is the set that denotes all possible bipartitions $\{A_{\alpha_i}|B_{\alpha_i}\}$ of the $n$ parties, $c(\alpha)$ is the cardinality of $\alpha$, and $\cP(\ket{\psi})$ is the product of all bipartite concurrences. Explicitly, 
\begin{align}\label{eq5}
    \cP(\ket{\psi})=&\prod\limits_{\alpha_i\in\alpha}\cC_{A_{\alpha_i}B_{\alpha_i}}(\ket{\psi})\,,\\
    c(\alpha)=&\left\{\begin{aligned}
    &\sum_{m=1}^{(n-1)/2}\Bigl(\!
  \begin{array}{c}
    n \\
    m \\
  \end{array}
\!\Bigr) ,\,&\text{if $n$ is odd}\,,\\
    &\sum_{m=1}^{(n-2)/2}\Bigl(\!
  \begin{array}{c}
    n \\
    m \\
  \end{array}
\!\Bigr)+\frac{1}{2}\Bigl(\!
  \begin{array}{c}
    n/2 \\
    n \\
  \end{array}
\!\Bigr),\,&\text{if $n$ is even}\,.
    \end{aligned}\right.\nonumber
\end{align}
\end{definition}
We can define the GBC for an $n$-partite mixed state $\rho$ via the convex roof extension,
\begin{equation}\label{eq6}
    \cG(\rho)=\underset{p_i,\ket{\psi_{i}}}{\text{min}}\sum_{i}{p_i\,\cG\bigl(\ket{\psi_{i}}\bigr)}\,,
\end{equation}
where $\{p_i,\ket{\psi_i}\}$ is an ensemble realization of $\rho$.
With this, we have the following theorem.

\begin{theorem}\label{theorem1}
For an arbitrary $n$-partite quantum state, the GBC is a GME measure.
\end{theorem}
\begin{proof}
First, we argue that GBC satisfies the key requirement for an entanglement measure, namely, the (strong) monotonicity, that GBC must not increase on average under LOCC operations. Since GBC is a concave function of entanglement monotones, the monotonicity is directly inherited; see Appendix~A for the detailed proof.

Next, we prove that GBC satisfies the two conditions $(\text{a})$ and $(\text{b})$ for a GME measure as in Definition~\ref{def1}. For a biseparable  pure state $\ket{\psi}_{\text{BS}}$, if it is separable between a set of party $A$ and its complement $\Bar{A}$, then the bipartite concurrence under the bipartition ${\bigl\{A\bigl|\Bar{A}\bigr\}\bigr.}$ vanishes, thus we have ${\cG\bigl(\ket{\psi}_{\text{BS}}\bigr)=0}$. Also, for a biseparable mixed state $\rho_{\text{BS}}$, the optimal ensemble via the convex roof construction in Eq.~\eqref{eq6} is composed of biseparable pure states, thus $\cG\bigl(\rho_{\text{BS}}\bigr)=0$, which proves condition $(\text{a})$. A GME pure state $\ket{\psi}_{\text{GE}}$ is not biseperable under any bipartition of the parties, then all the bipartite concurrences are strictly positive, thus $\cG\bigl(\ket{\psi}_{\text{GE}}\bigr)>0$. Also for genuinely entangled mixed state $\rho_{\text{GE}}$, there does not exist an ensemble realization of $\rho_{\text{GE}}$ where all the pure states are biseparable, thus $\cG\bigl(\rho_{\text{GE}}\bigr)>0$, which proves condition $(\text{b})$.
\end{proof}

\begin{table*}[t]
\begin{center}
\caption{Properties of various GME measures. The concurrence fill \cite{PhysRevLett.127.040403} is defined for three-qubit quantum states only, thus fails the property of scalability. It is unknown for which state GGM \cite{PhysRevA.77.062304} takes its maximum value, so the normalization of GGM needs further investigation. There are cases when GGM and GMC \cite{PhysRevA.83.062325} exhibit sharp peaks for the family of states varying continuously, while the GBC is always smooth. The GBC is well-defined in terms of all the listed properties.}
    \begin{tabular}{|c|c|c|c|c|c|c|}
     \hline
     & ~Discriminance~ & ~Monotonicity~ & ~Convexity~ & ~Scalability~ & ~Normalization~ & ~Smoothness~ \\ \hline
    GGM & \checkmark & \checkmark & \checkmark & \checkmark & ? & \ding{53} \\\hline
    GMC & \checkmark & \checkmark & \checkmark & \checkmark & \checkmark & \ding{53} \\\hline
    ~Concurrence fill~ & \checkmark & \checkmark & \checkmark & \ding{53} & \checkmark & \checkmark \\\hline
    GBC & \checkmark & \checkmark & \checkmark & \checkmark & \checkmark & \checkmark \\
    \hline
\end{tabular}\label{tb:PropsofGME}
\end{center}
\end{table*}

A direct observation of Definition~\ref{def2} is that when comparing states in the same Hilbert space, which is what we're concerned with in this work, the regularization in Eq.~\eqref{eq2} will not influence the relative magnitude of GBC. The GBC takes its maximal value of $1$ for absolutely maximally entangled (AME) states~\cite{PhysRevA.86.052335}, and thus can be viewed as a \emph{proper} GME measure for the following reasons. Since the smaller subsystem of the AME states under arbitrary bipartitions is always maximally mixed, they are natural candidates for developing quantum secret sharing schemes~\cite{PhysRevA.86.052335} and quantum error correction codes~\cite{PhysRevA.69.052330}. Specifically, the three-qubit AME state, i.e., the GHZ state is maximally entangled in the sense that it can perform perfect quantum teleportation, while up to now the $W$ state, though being the maximally entangled state in the $W$ class, can only achieve a maximal success rate of $2/3$~\cite{Joo_2003}. Although AME states don't always exist~\cite{PhysRevLett.118.200502}, it is proven that when the local dimensions are chosen big enough, there are AME states for any number of parties~\cite{helwig2013absolutely}. Thus we have grounds for saying that a proper entanglement measure should be maximized for the AME states, if they exist.

\section{Properties of GBC}%
The axiomatic approach for constructing entanglement measures suggests several advantageous properties~\cite{Geometry}. Here we examine the following six main properties for GME measures:
(1) Discriminance, which denotes the conditions (a) and (b) in Definition~\ref{def1} for GME measures.
(2) (Strong) monotonicity, that the entanglement cannot increase on average under LOCC operations
\begin{equation}\label{eq7}
    E(\rho) \ge \sum_i{p_i E\bigl(\sigma_i\bigr)}\,,
\end{equation}
where $\{p_i, \sigma_i\}$ is an ensemble produced by an arbitrary LOCC operation $\Phi_{\text{LOCC}}$, i.e.,
\begin{equation}
    \Phi_{\text{LOCC}}(\rho)=\sum_i p_i\sigma_i\,.
\end{equation}
(3) Convexity, that by mixing different quantum states the total amount of entanglement cannot increase
\begin{equation}
    E\bigl(\lambda\rho+(1-\lambda)\sigma\bigr)\le\lambda E(\rho)+(1-\lambda)E(\sigma)\,.
\end{equation}
If a measure is defined for pure states, then it can always be convex after the convex roof extension to mixed states.
(4) Scalability, that the measure is defined for arbitrary multipartite states, regardless of the local dimensions or the number of parties. 
(5) Normalization, that the measure is maximized for AME states. An interesting point is that it is still unknown for which state GGM can be maximized. In addition, we argue that GBC is more natural than GMC and GGM in view of (6) smoothness, that the measure does not create sharp peaks when measuring continuously varying pure states~\cite{PhysRevLett.127.040403}. Smoothness is absent for measures defined with the minimal argument like GGM and GMC, while the analytical expression of GBC guarantees that it is smooth. See Table~\ref{tb:PropsofGME} for a summary.

\section{Entanglement ordering with GBC}%
As an entanglement monotone, GBC obtains the infallible entanglement ordering for any pair of quantum states that can transform from one to the other with deterministic LOCC. However, even states in the same SLOCC equivalent class can be incomparable with monotonicity~\cite{PhysRevLett.83.436}, thus different entanglement monotones lead to much varied results of entanglement ordering. In this section, we give illustrative examples of the GBC's applications when measuring multipartite entanglement and show that GBC possesses advantages over other GME measures in measuring certain types of entanglement.

We first investigate two typical permutation symmetric states that belong to different SLOCC equivalent classes: the GHZ states and the $W$ states~\cite{PhysRevA.62.062314}. The three-qubit GHZ states are more capable for quantum teleportation than the three-qubit $W$ states~\cite{Joo_2003}. Also, the $W$ states can be approximated by states in the GHZ class to arbitrary precision, but the reverse is not true~\cite{walter2017multipartiteentanglement}, thus the GHZ states can be regarded as more entangled than the $W$ states. Through analytical derivation and calculation, we find that for multi-qubit systems the GBC of GHZ states is strictly larger than that of $W$ states, which is consistent with the above observations. In addition, when the number of qubits approaches infinity, the ratio of GBC of GHZ states over the GBC of $W$ states approaches 1 asymptotically. See Appendix~B for the explicit derivations.

The GBC leads to different entanglement orderings as compared to other GME measures.  As an example we consider the following two types of three-qubit states
\begin{equation}\label{eq10}
\begin{aligned}
    \ket{\psi_A}&=\frac{1}{\sqrt{2}}\bigl(\cos\theta\ket{000}+\sin\theta\ket{100}\bigr)+\frac{1}{\sqrt{2}}\ket{111},\\
     \ket{\psi_B}&=\cos\theta\ket{000}+\sin\theta\ket{111},
\end{aligned}
\end{equation}
with ${\theta\in[0,\pi/2]}$. The GMC and GGM get exactly the same entanglement ordering for three-qubit states and can therefore be treated as a whole~\cite{PhysRevLett.127.040403}. As shown in Figs.~\ref{fig:concfill} and \ref{fig:GMC}, there exist pairs of a type-A state and type-B state that are indistinguishable with concurrence fill or GMC but contains different amounts of GBC, and vice versa. This can be shown by drawing vertical or horizontal lines, then compare the states at the intersection point. In addition, by introducing horizontal and vertical lines from a given typ-A state, we can specify a set of type-B states between the intersection points possessing less(more) GBC but more(less) concurrence fill(GMC) when compared with the type-A state. Similar results can be drawn by inspecting from a type-B state. 

\begin{figure}[t]
    \includegraphics[width=.95\columnwidth]{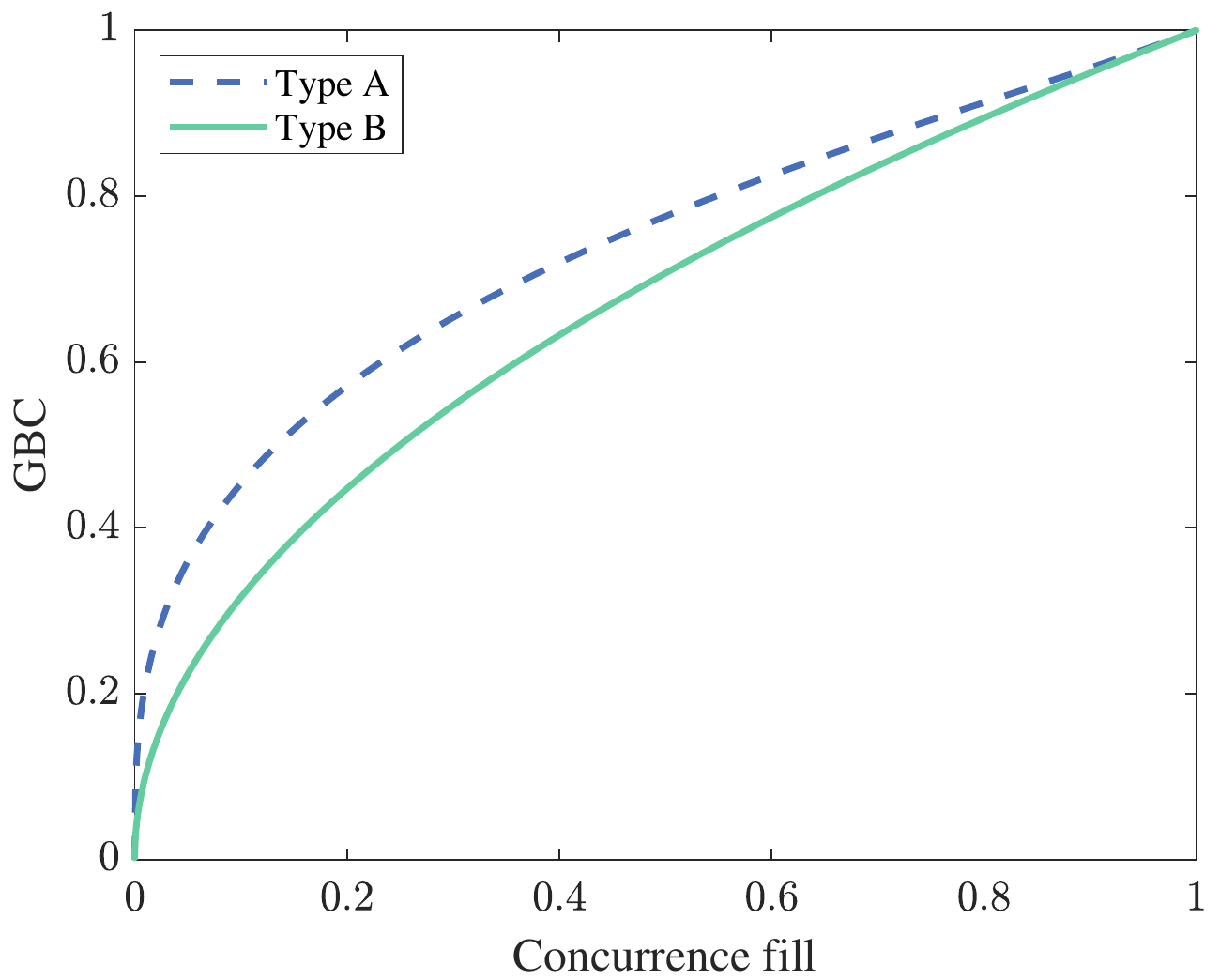}
    \caption{The GBC and concurrence fill for two different types of quantum states as in Eq.~\eqref{eq10}. Each point on the curve represents a three-qubit state. Given a type A state, we can find a set of type B states that have larger concurrence fill but smaller GBC.}
    \label{fig:concfill}
\end{figure}
\begin{figure}[t]
    \includegraphics[width=.95\columnwidth]{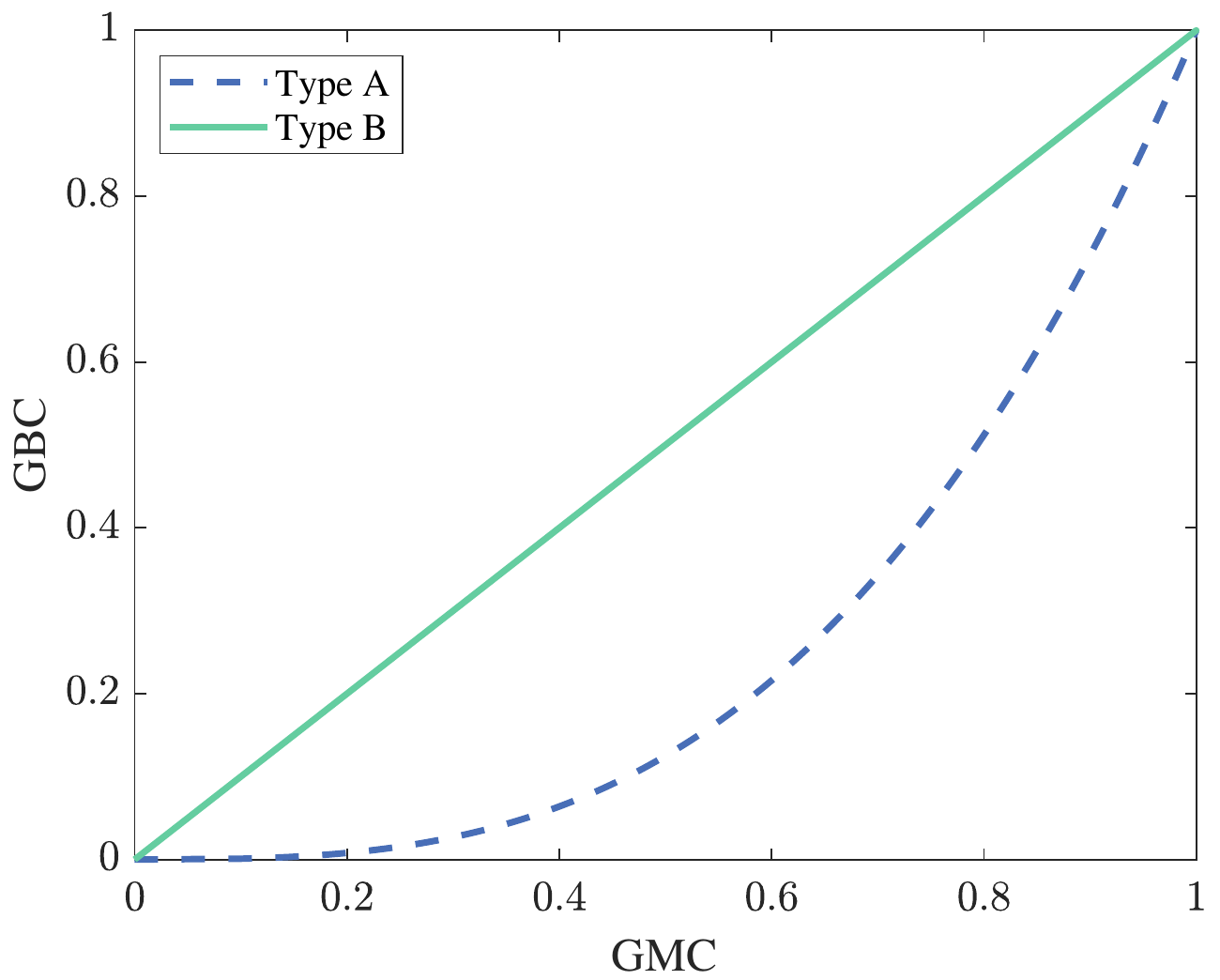}
    \caption{The GBC and GMC for two different types of quantum states as in Eq.~\eqref{eq10}. Each point on the curve represents a three-qubit state. Given a type-A state, we can find a set of type-B states that have smaller GMC but larger GBC. The same results are applicable to GGM, since for three-qubit states GGM and GBC are equivalent in entanglement ordering.}
    \label{fig:GMC}
\end{figure}
\begin{figure}[t]
    \includegraphics[width=.95\columnwidth]{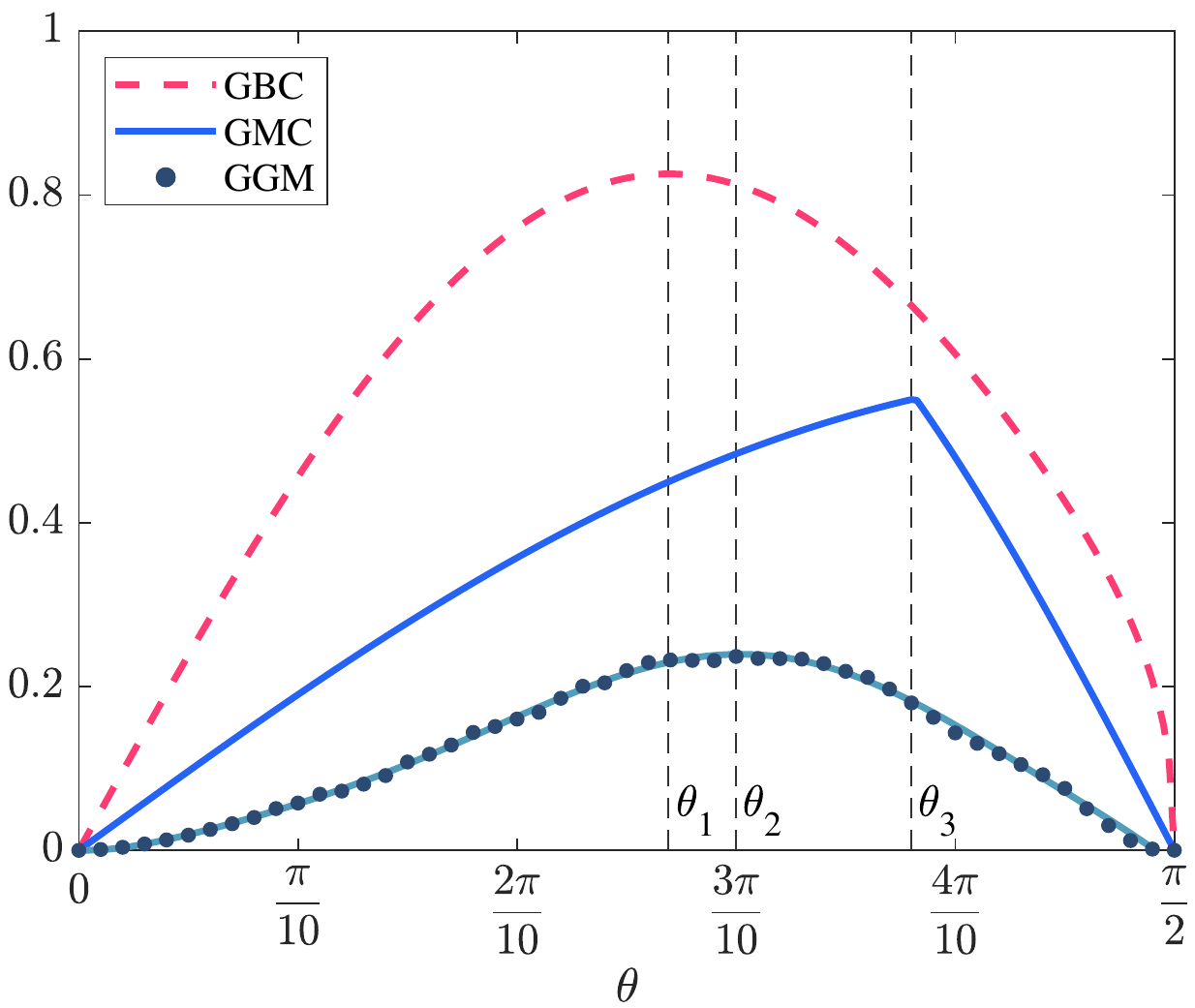}
    \caption{The GBC, GMC and GGM for the four-qubit type C states as in Eq.~\eqref{eq11}. Each point on the curve represents a type-C state with the corresponding $\theta$. The peak of GBC is at $\theta_1$, while the GMC has a sharp peak at $\theta_3$. It can be seen that when $\theta$ increases from $\theta_1$ to $\theta_3$, the GBC is decreasing while the GMC is increasing. We compute the GGM at 51 values of $\theta$ from $0$ to $\pi/2$ with numerical methods, and a peak $\theta_2$ inferred from the fitting curve is larger than $\theta_1$.}
    \label{fig:eg4Q}
\end{figure}

For another example, we consider the following family of four-qubit states
\begin{equation}\label{eq11}
\begin{aligned}
    \ket{\psi_C}=&\sin\theta\biggl(\cos\Bigl(\frac{3\pi}{5}\Bigr)\ket{0100}+\sin\Bigl(\frac{3\pi}{5}\Bigr)\ket{1000}\biggr)\\
    &+\cos\theta\ket{0011},
\end{aligned}
\end{equation}
with ${\theta\in[0,\pi/2]}$. As shown in Fig.~\ref{fig:eg4Q}, as $\theta$ increases from $\theta_1$ to $\theta_3$ the GBC is decreasing with $\theta$ while the GMC is increasing, thus for two arbitrary states in this range the entanglement ordering is opposite for GBC and GMC. We use the Gilbert algorithm \cite{PhysRevLett.120.050506,Gilbert} to evaluate the GGM based on quantum relative entropy at 51 values of $\theta$ sampled uniformly from $0$ to $\pi/2$, and a peak at $\theta_2$ can be inferred from the fitting curve, which is between the peaks of GBC and GGM. Therefore the entanglement ordering for type C states with $\theta\in[\theta_1,\theta_2]$ is unique for GBC. By looking at a pair of states with the same GMC or GGM, we are likely to find that they possess different GBC, meaning that GBC is able to detect the difference in entanglement while GMC or GGM fails to do so. Moreover, Fig.~\ref{fig:eg4Q} also demonstrates an example of smoothness. A sharp peak of GMC emerges with the varying $\theta$, while the GBC is smooth with a continuous slope.

\section{Summary}%
In this work we have proposed a well-defined GME measure, the GBC, which applies to arbitrary multipartite states. We have shown that GBC is qualified for quantifying potential quantum resources, benefiting from desirable properties such as monotonicity and smoothness. By utilizing the symmetry in the definition of GBC, analytical results were obtained for $n$-qubit GHZ states and $W$ states, and the fact that GBC regards the GHZ state as more entangled than the $W$ state was shown to be proper in specific scenarios. With further explicit examples, we have demonstrated that GBC can distinguish the difference in certain features of GME while other GME measures cannot. Hence, GBC could serve as a complement to current GME measures. For future work, it's interesting to question whether experimentally implementable methods could be employed to estimate a lower bound for GBC~\cite{PhysRevApplied.13.054022}, and to find more potential applications of GBC in tasks involving GME.

\acknowledgments
We are grateful to Ye-Chao Liu, Hao-Sen Chen, and Yixuan Hu for helpful discussions.
This work was supported by the National Natural Science Foundation of China (Grants No.~12175014 and No.~11805010).

\section*{Appendix A: Proof of the monotonicity of GBC}\label{app:properties}
In this appendix, we prove explicitly that GBC is an entanglement monotone, thus qualifies for the task of entanglement quantification. 

Bipartite concurrence satisfies the property of strong monotonicity~\cite{PhysRevA.74.052303}, i.e.,
\begin{equation}\label{eq14}
    \sum_i{p_i\cC_{AB}\bigl(\sigma_i\bigr)}\le\cC_{AB}(\rho)\,,
\end{equation}
where $\{p_i, \sigma_i\}$ is an ensemble produced by an arbitrary LOCC channel $\Phi_{\text{LOCC}}$ acting on the multipartite quantum state $\rho$, such that
\begin{equation}\label{eq15}
    \Phi_{\text{LOCC}}(\rho)=\sum_i{p_i\sigma_i}\,.
\end{equation}
Note that $\Phi_{\text{LOCC}}$ is also an LOCC channel with respect to the bipartition $\{A|B\}$.
First we examine the monotonicity for pure states, i.e., when $\rho$ and $\sigma_i$ are all pure states. Since all the bipartite concurrences in the definition of GBC satisfy Eq.~\eqref{eq14}, we have
\begin{equation}
\begin{aligned}
    \cG(\rho)=&\sqrt[{c(\alpha)}]{\cP(\rho)}\ge\sqrt[{c(\alpha)}]{\prod\limits_{\alpha_j\in\alpha}\Bigl(\sum_i{p_i\cC_{A_{\alpha_j}B_{\alpha_j}}\bigl(\sigma_i\bigr)\Bigr)}}\\
    \ge&\sum_i{p_i\Biggl(\sqrt[{c(\alpha)}]{\prod\limits_{\alpha_j\in\alpha}\cC_{A_{\alpha_j}B_{\alpha_j}}\bigl(\sigma_i\bigr)}\Biggr)}=\sum_i{p_i\cG\bigl(\sigma_i\bigr)}\,,
\end{aligned}
\end{equation}
where the concavity of the geometric mean function defined on the convex set with all the variables ranging from 0 to 1 is employed, which can be easily proven with the Mahler's inequality. 

Then, the monotonicity of mixed states can be naturally inherited from the monotonicity of pure states via the convex roof extension, as shown in Ref.~\cite{doi:10.1080/09500340008244048}. All measures constructed via the convex roof extension are convex~\cite{RevModPhys.81.865}. Furthermore, local unitary (LU) invariance is a necessary condition for monotonicity, thus GBC also satisfies the property of LU invariance.

\section*{Appendix B: The GBC of $n$-qubit GHZ states and $W$ states}\label{app:GHZandW}
The general expression of an $n$-qubit GHZ state is
\begin{equation}
\begin{aligned}
    \ket{\text{GHZ}_n}&=\frac{1}{\sqrt{2}}\sum_{j=0}^{1}\ket{j}^{\otimes n}\,,\\
    \rho_{\text{GHZ}}&=\ket{\text{GHZ}_n}\bra{\text{GHZ}_n}=\frac{1}{2}\sum_{i,j=0}^{1}\ket{j}^{\otimes n}\bra{i}^{\otimes n}\,,
\end{aligned}
\end{equation}
of which the bipartite concurrence is
\begin{equation}
\begin{aligned}
    \cC_{AB}(\rho_\text{GHZ})&=\sqrt{\frac{d_\text{min}}{d_\text{min}-1}\bigl(1-\tr\bigl(\rho_A^2\bigr)\bigr)}\\
    &=\sqrt{\frac{d_\text{min}}{2(d_\text{min}-1)}}
\end{aligned}
\end{equation}
for an arbitrary bipartition $\bigl\{A\bigl|B\bigr\}\bigr.$ with $d_\text{min}$ denoting the dimension of the smaller subsystem. The product of all bipartite concurrences is
\begin{widetext}
 \begin{eqnarray} 
\begin{aligned}
    \cP\bigl(\ket{\text{GHZ}_n}\bigr)=&\left\{\begin{aligned}
    &\prod_{m=1}^{\frac{n-1}{2}}{\sqrt{\frac{2^m}{2(2^m-1)}}^{\tbinom{n}{m}}} \,,\,&\text{if $n$ is odd}\,,\\
    &\sqrt{\frac{2^{n/2}}{2(2^{n/2}-1)}}^{\frac{1}{2}\tbinom{n}{n/2}}\prod_{m=1}^{\frac{n-2}{2}}{\sqrt{\frac{2^m}{2(2^m-1)}}^{\tbinom{n}{m}}},\,&\text{if $n$ is even}\,.
    \end{aligned}\right.
\end{aligned}
  \end{eqnarray}
\end{widetext} 
Then the GBC for $\ket{\text{GHZ}_n}$ is given by
\begin{equation}
    \cG\bigl(\ket{\text{GHZ}_n}\bigr)=\sqrt[{c(\alpha)}]{\cP\bigl(\ket{\text{GHZ}_n}\bigr)}\,,
\end{equation}
where $c(\alpha)$ denotes the cardinality of $\alpha$, which is the set of all possible bipartitions.

The general expression of an $n$-qubit $W$ state is
\begin{equation}
\begin{aligned}
    \ket{W_n}&=\ket{n-1,1}\,,\\
    \rho_{W}&=\ket{W_n}\bra{W_n}=\ket{n-1,1}\bra{n-1,1}\,,    
\end{aligned}
\end{equation}
where $\ket{n-1,1}$ is the normalized fully symmetric state with $n-1$ zeros and $1$ one. For example, the three-qubit $W$ state is
\begin{equation}
    \ket{W_3}=\frac{1}{\sqrt{3}}\bigl(\ket{100}+\ket{010}+\ket{001}\bigr)\,.
\end{equation}
First we compute the one-to-other concurrence of $\ket{W_n}$ by tracing over an arbitrary qubit, such that
the reduced density matrix is 
\begin{equation}
\begin{aligned}
    \rho_1=&\braket{0|W_n}\braket{W_n|0}+\braket{1|W_n}\braket{W_n|1}\\
    =&\frac{n-1}{n}\ket{n-2,1}\bra{n-2,1}+\frac{1}{n}\ket{n-1,0}\bra{n-1,0}\,,
\end{aligned}    
\end{equation}
so the one-to-other concurrence is
\begin{equation}\label{eq22}
    \cC_{1,n-1}=\sqrt{2\bigl(1-\tr\bigl(\rho_1^2\bigr)\bigr)}
    =2\sqrt{\frac{n-1}{n^2}}\,.
\end{equation}
Similarly, the two-to-other concurrence can be computed by tracing over another qubit, such that the reduced density matrix is 
\begin{equation}
\begin{aligned}
    \rho_2=&\braket{0|\rho_1|0}+\braket{1|\rho_1|1}\\
    =&\frac{n-2}{n}\ket{n-3,1}\bra{n-3,1}+\frac{2}{n}\ket{n-2,0}\bra{n-2,0}\,,
\end{aligned}    
\end{equation}
so the two-to-other concurrence is
\begin{equation}
    \cC_{2,n-2}=\sqrt{\frac{4}{3}\bigl(1-\tr\bigl(\rho_2^2\bigr)\bigr)}
    =4\sqrt{\frac{n-2}{3n^2}}\,.
\end{equation}
By induction, the $m$-to-other concurrence is
\begin{equation}
    \cC_{m,n-m}=\sqrt{\frac{2^m}{2^m-1}\bigl(1-\tr\bigl(\rho_m^2\bigr)\bigr)}
    =\sqrt{\frac{(mn-m^2)2^{m+1}}{(2^m-1) n^2}}\,,
\end{equation}
where $m\le\frac{n}{2}$. Then the product of all the bipartite concurrences can be computed
\begin{widetext}
 \begin{eqnarray} 
\begin{aligned}
    \cP\bigl(\ket{\text{W}_n}\bigr)=&\left\{\begin{aligned}
    &\prod_{m=1}^{\frac{n-1}{2}}{\sqrt{\frac{(mn-m^2)2^{m+1}}{(2^m-1) n^2}}^{\tbinom{n}{m}}} \,,\,&\text{if $n$ is odd}\,,\\
    &\sqrt{\frac{2^{n/2}}{2(2^{n/2}-1)}}^{\frac{1}{2}\tbinom{n}{n/2}}\prod_{m=1}^{\frac{n-2}{2}}{\sqrt{\frac{(mn-m^2)2^{m+1}}{(2^m-1) n^2}}^{\tbinom{n}{m}}},\,&\text{if $n$ is even}\,.
    \end{aligned}\right.
\end{aligned}
  \end{eqnarray}
\end{widetext} 
Finally, the GBC for $\ket{W_n}$ is given by 
\begin{equation}
    \cG\bigl(\ket{W_n}\bigr)=\sqrt[c(\alpha)]{\cP\bigl(\ket{W_n}\bigr)}\,,
\end{equation}
with $c(\alpha)$ denoting the cardinality of $\alpha$, which is the set of all possible bipartitions.

Now we prove that when the number of qubits $n$ approaches infinity, the GBCs of the GHZ states and the $W$ states tend to be the same.
By looking at the ratio
\begin{equation}
\begin{aligned}
    \frac{\cP\bigl(\ket{W_n}\bigr)}{\cP\bigl(\ket{\text{GHZ}_n}\bigr)}=\prod_{m=1}^{\lfloor\frac{n}{2}\rfloor}{2\sqrt{\frac{mn-m^2}{n^2}}^{\tbinom{n}{m}}},
\end{aligned}
\end{equation}
we have
\begin{equation}
    \begin{aligned}
    \lim_{n\to+\infty} \frac{\cG\bigl(\ket{W_{2n+1}}\bigr)}{\cG\bigl(\ket{\text{GHZ}_{2n+1}}\bigr)}&=\sqrt[c(\alpha)]{\frac{\cP\bigl(\ket{W_{2n+1}}\bigr)}{\cP\bigl(\ket{\text{GHZ}_{2n+1}}\bigr)}}\\
    &\ge\lim_{n\to+\infty}\frac{\sum_{m=1}^{n}{\tbinom{2n+1}{m}}}{\sum_{m=1}^{n}\frac{1}{2}{\tbinom{2n+1}{m}}\sqrt{\frac{(2n+1)^2}{m(2n+1)-m^2}}}\\
    &=\lim_{n\to+\infty}\frac{{\tbinom{2n+1}{n}}}{\frac{1}{2}{\tbinom{2n+1}{n}}\sqrt{\frac{(2n+1)^2}{n(n+1)}}}\\
    &=1^-\,,
    \end{aligned}
\end{equation}
\begin{equation}
    \begin{aligned}
    \lim_{n\to+\infty} \frac{\cG\bigl(\ket{W_{2n}}\bigr)}{\cG\bigl(\ket{\text{GHZ}_{2n}}\bigr)}&=\sqrt[c(\alpha)]{\frac{\cP\bigl(\ket{W_{2n}}\bigr)}{\cP\bigl(\ket{\text{GHZ}_{2n}}\bigr)}}\\
    &\ge\lim_{n\to+\infty}\frac{\sum_{m=1}^{n-1}{\tbinom{2n}{m}}+\frac{1}{2}\tbinom{2n}{n}}{\sum_{m=1}^{n-1}\frac{1}{2}{\tbinom{2n}{m}}\sqrt{\frac{2n^2}{2mn-m^2}}+\frac{1}{2}\tbinom{2n}{n}}\\
    &=1^-\,,
    \end{aligned}
\end{equation}
\begin{figure}[t]
    \includegraphics[width=1\columnwidth]{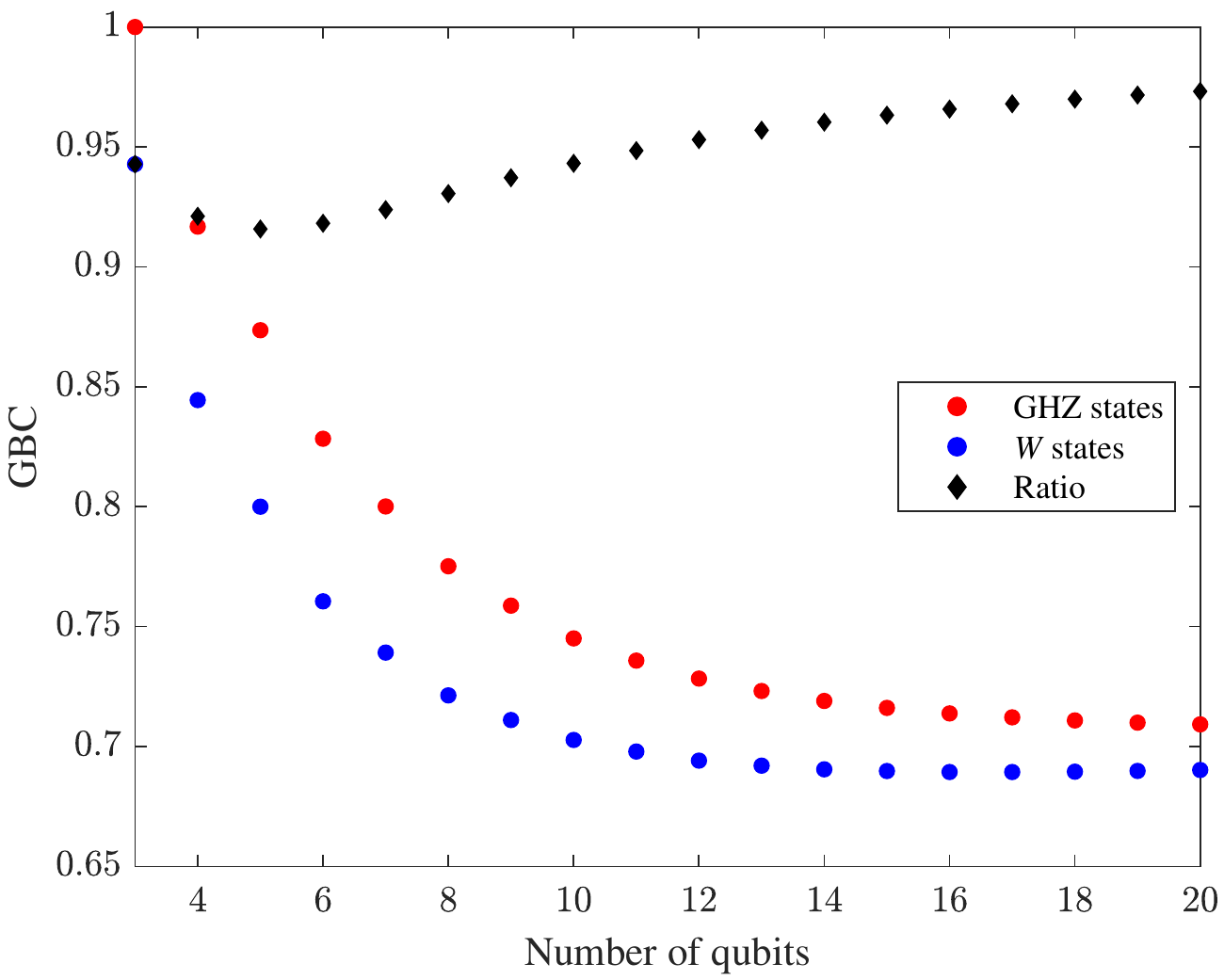}
    \caption{The GBC of $n$-qubit GHZ states and $W$ states. It can be seen that in a multi-qubit system the GHZ state is always more entangled than the $W$ state in terms of GBC, and the ratio of ${\cG\bigl(\ket{W_n}\bigr)}$ over ${\cG\bigl(\ket{\text{GHZ}_n}\bigr)}$ approaches 1 asymptotically.}
    \label{fig:GHZandW}
\end{figure}
where we have used the harmonic mean-geometric mean inequality and a variant of the Stolz–Ces\`aro theorem.
Thus
\begin{equation}
    \lim_{n\to+\infty} \frac{\cG\bigl(\ket{W_n}\bigr)}{\cG\bigl(\ket{\text{GHZ}_n}\bigr)}\ge1^-\,.
\end{equation}
Since $\cG\bigl(\ket{W_n}\bigr)$ is always smaller than $\cG\bigl(\ket{\text{GHZ}_n}\bigr)$, we have
\begin{equation}
    \lim_{n\to+\infty} \frac{\cG\bigl(\ket{W_n}\bigr)}{\cG\bigl(\ket{\text{GHZ}_n}\bigr)}=1^-\,.
\end{equation}
See Fig.~\ref{fig:GHZandW} for a visualization of the results.


%


\end{document}